\documentclass[pdflatex, sn-mathphys]{sn-jnl}

\usepackage[T1]{fontenc}
\usepackage{enumitem}
\usepackage{amsfonts}
\usepackage{amsmath}
\usepackage[capitalize]{cleveref}

\jyear{2023}

\newtheorem{theorem}{Theorem}[section]
\newtheorem{lemma}[theorem]{Lemma}
\newtheorem{claim}[theorem]{Claim}
\newtheorem{corollary}[theorem]{Corollary}
\newtheorem{fact}[theorem]{Fact}
\theoremstyle{definition}
\newtheorem{definition}[theorem]{Definition}
\theoremstyle{remark}
\newtheorem{remark}[theorem]{Remark}

\renewcommand{\S}{\mathcal{S}}
\renewcommand{\A}{\mathcal{A}}
\renewcommand{\B}{\mathcal{B}}
\newcommand{\I}{\mathcal{I}}

\newcommand{\deltalog}{\delta\log\tfrac{n}{\delta}}
\newcommand{\deltalos}{\delta\log\tfrac{n\log \sigma}{\delta \log n}}
\newcommand{\gammalog}{\gamma\log\frac{n}{\gamma}}

\newcommand{\Zz}{\mathbb{Z}_{\ge 0}}
\newcommand{\floor}[1]{\lfloor #1 \rfloor}
\newcommand{\ceil}[1]{\lceil #1 \rceil}

\renewcommand{\exp}[1]{\mathtt{exp}(#1)}

\newcommand{\dd }{\mathinner{.\,.}}
\newcommand{\Oh}{O}

\begin{document}

\title[Near-Optimal Search Time in \texorpdfstring{$\delta$}{delta}-Optimal Space
]{Near-Optimal Search Time in \texorpdfstring{$\delta$}{delta}-Optimal Space, and Vice Versa\footnote{Funded in part by Basal Funds FB0001,  Fondecyt Grant 1-200038, and Ph.D Scholarship 21210579, ANID, Chile.}}

\author[1]{\fnm{Tomasz} \sur{Kociumaka}}
\author[2]{\fnm{Gonzalo} \sur{Navarro}}
\author*[2]{\fnm{Francisco} \sur{Olivares}}\email{folivares@uchile.cl}

\affil[1]{\orgname{Max Planck Institute for Informatics}, \orgaddress{\city{Saarbr\"ucken}, \country{Germany}}}

\affil[2]{\orgdiv{CeBiB --- Centre for Biotechnology and Bioengineering\\ Department of Computer Science}, \orgname{University of Chile},\orgaddress{ \country{Chile}}}

\abstract{Two recent lower bounds on the compressibility of repetitive sequences, $\delta \le \gamma$, have received much attention.
It has been shown that a length-$n$ string $S$ over an alphabet of size $\sigma$ can be represented within the optimal $\Oh(\deltalos)$ space, and further, that within that space one can find all the $occ$ occurrences in $S$ of any length-$m$ pattern in time $O(m\log n + occ \log^\epsilon n)$ for any constant $\epsilon>0$.
Instead, the near-optimal search time $O(m+({occ+1})\log^\epsilon n)$ has been achieved only within $O(\gammalog)$ space.
Both results are based on considerably different locally consistent parsing techniques.
The question of whether the better search time could be supported within the $\delta$-optimal space remained open.
In this paper, we prove that both techniques can indeed be combined to obtain the best of both worlds: $O(m+({occ+1})\log^\epsilon n)$ search time within $\Oh(\deltalos)$ space. 
Moreover, the number of occurrences can be computed in $O(m+\log^{2+\epsilon}n)$ time within $\Oh(\deltalos)$ space.
 We also show that an extra sublogarithmic factor on top of this space enables optimal $O(m+occ)$ search time, whereas an extra logarithmic factor 
 enables optimal $O(m)$ counting time.}


\maketitle

\section{Introduction}

The amount of data we are expected to handle has been growing steadily in the last decades~\cite{Plos15}.
The fact that much of the fastest-growing data is composed of highly repetitive sequences has raised interest in text indexes whose size can be bounded by some measure of repetitiveness~\cite{indexing2} and in the study of those repetitiveness measures~\cite{indexing}.
Since statistical compression does not capture repetitiveness well~\cite{on_compressing}, various other measures have been proposed for this case.
Two recent ones, which have received much attention because of their desirable properties, are the size $\gamma$ of the smallest string attractor~\cite{attractors} and a function $\delta$ of the substring complexity~\cite{talg,delta}.
Every string satisfies $\delta \le \gamma$~\cite{talg} (with $\delta = o(\gamma)$ in some string families~\cite{delta}), and $\gamma$ asymptotically lower-bounds many other measures sensitive to repetitiveness~\cite{attractors} (e.g., the size of the smallest Lempel--Ziv parse~\cite{lempel-ziv}).
On the other hand, any string $S\in [0\dd \sigma)^n$ can be represented within $O(\deltalos)$ space, and this bound is tight for the full spectrum of parameters $n$, $\sigma$, and $\delta$~\cite{delta}.

A more ambitious goal than merely representing $S$ in compressed space is to \emph{index} it within that space so that, given any pattern $P$, one can efficiently find all the $occ$ occurrences of $P$ in $S$. 
Interestingly, it has been shown that, for any constant $\epsilon > 0$, one can index $S$ within the optimal $O(\deltalos)$ space and then report the occurrences of any length-$m$ pattern in $O(m\log n + occ \log^\epsilon n)$ time~\cite{delta}.
If one allows the higher $O(\gammalog)$ space, the search time can be reduced to $O(m+(occ+1)\log^\epsilon n)$~\cite{talg}, which is optimal in terms of the pattern length and near-optimal in the time per reported occurrence.
Slightly more space, $O(\gammalog\log^\epsilon n)$, allows for a truly optimal search time, $O(m+occ)$.\footnote{In this work, we assume that $P[1..m]$ is represented in $O(m)$ space. For small alphabets, the packed setting, where $P$ occupies $O(\ceil{\frac{m\log \sigma}{\log n}})$ space, could also be considered; see \cite[Sec.~2.2]{indexing}.}

The challenge of obtaining the near-optimal $O(m+({occ+1})\log^\epsilon n)$ search time within optimal $\Oh(\deltalos)$ space was posed~\cite{talg,delta}, and this is what we settle on the affirmative in this paper.
Both previous results build convenient context-free grammars on $S$ and then adapt a classical grammar-based index on it~\cite{CNspire12.1,CNP20}.
The index based on attractors~\cite{talg} constructs a grammar from a locally consistent parsing~\cite{MSU97} of $S$ that forms blocks in $S$ ending at every local minimum with respect to a randomized ordering of the alphabet, collapsing every block into a nonterminal and iterating.
The smaller grammar based on substring complexity~\cite{delta} uses another locally consistent parsing obtained by \emph{recompression}~\cite{jez}, which randomly divides the alphabet into ``left'' and ``right'' symbols and combines every left-right pair into a nonterminal, also iterating.
The key to obtaining $\delta$-bounded space is to pause the pairing on symbols whose expansions become too long for the iteration where they were formed~\cite{delta}.
We show that the pausing idea can be applied to the first kind of locally consistent grammar as well so that it yields the desired time and space complexities.
The next theorem summarizes our first result.

\begin{theorem}\label{thm:main}
For every constant $\epsilon>0$, given a string $S\in [0\dd \sigma)^n$ with measure $\delta$, one can build in $O(n)$ expected time a data structure using $\Oh(\deltalos)$ words of space such that, later, given a pattern $P[1\dd m]$, one can find all of its $occ$ occurrences in $S$
in $\Oh(m+\log^\epsilon\delta + occ\log^\epsilon (\deltalos)) \subseteq O(m+(occ+1)\log^\epsilon n)$ time.
\end{theorem}

Apart from this near-optimal search time within optimal space, we can obtain optimal search time within near-optimal space, increasing the space by only a sublogarithmic factor.

\begin{theorem}\label{thm:main2}
For every constant $\epsilon>0$, given a string $S\in [0\dd \sigma)^n$ with measure $\delta$, one can build in $\Oh(n+\deltalos\log n)$ expected time a data structure using $\Oh(\deltalos\log^\epsilon (\deltalos)) \subseteq \Oh(\deltalos\log^\epsilon n)$ words of space such that, later, given a pattern $P[1\dd m]$, one can find all of its $occ$ occurrences in $S$ in optimal time $O(m+ occ)$.
\end{theorem}

Finally, we show how to efficiently count the number of occurrences of a pattern in $S$ within $\delta$-bounded space, while matching the times of the previous $\gamma$-bounded space structure.

\begin{theorem}\label{thm:main3}
For every constant $\epsilon>0$, given a string $S\in [0\dd \sigma)^n$ with measure $\delta$, one can build in $O(n\log n)$ expected time a data structure using $\Oh(\deltalos)$ words of space such that, later, given a pattern $P[1\dd m]$, one can compute the number $occ$ of its occurrences in $S$ in time $O(m+ \log^{2+\epsilon} n)$. We can also build, in $O(n\log n)$ expected time, a data structure using $\Oh(\deltalos\log (\deltalos))\subseteq \Oh(\deltalos\log n)$ words of space that computes $occ$ in optimal time $O(m)$.
\end{theorem}

Our algorithms are designed for the (standard) RAM model of computation with machine words of $w = \Theta(\log n)$ bits each. 
By default, we measure the space in words, which means that $O(x)$ space comprises $O(x\log n)$ bits.

A conference version of this paper appeared in \emph{Proc. LATIN 2022}~\cite{KNOlatin22}. The current version provides more refined space results, which incorporate the alphabet size $\sigma$ and reach optimality for every $n$, $\delta$, and $\sigma$~\cite{delta}. It also includes the results on counting the number of occurrences and achieves optimal search times within near-optimal space.

\section{Notation and Basic Concepts}\label{sec:notation}

A string is a sequence $S[1 \dd n] = S[1]\cdot S[2]\cdots S[n]$ of symbols, where each
symbol belongs to an alphabet $\Sigma = [0\dd \sigma)=\{0, \ldots, \sigma-1\}$. 
We denote as $\Sigma(S)$ the subset of $\Sigma$ consisting of symbols that occur in $S$.
The length of $S$ is denoted $\vert S\vert  = n$. The only string of length zero is denoted $\varepsilon$, so $\vert \varepsilon \vert = 0$. We assume that the
alphabet size is a polynomial function of $n$, that is, $\sigma = n^{O(1)}$.
The concatenation of strings $S$ and $S'$ is denoted
$S\cdot S' = SS'$.
A string $S'$ is a substring of $S$ if $S' = S(i\dd j] = S[i+1]\cdots S[j]$ for some $0 \leq i \leq j \leq n$; in particular $S(i \dd j] = \varepsilon$ if $j=i$.
With the term \emph{fragment}, we refer to a particular occurrence $S(i\dd j]$ of a substring in $S$ (not just the substring content). 
We also use other combinations of round and square brackets to denote fragments:
 $S(i\dd j) = S(i \dd j - 1]$,  $S[i\dd j] = S(i-1 \dd j]$, and $S[i\dd j) = S(i-1 \dd j - 1]$.
We use $S^{rev}$ to denote the reverse of $S$, that is, $S^{rev} = S[n] \cdot S[n-1] \cdots S[1]$. 

A \emph{straight line program} (\textsc{SLP}) is a
context-free grammar where each nonterminal appears once at the left-hand side
of a rule, and where the nonterminals can be sorted so that the right-hand sides refer to terminals and preceding nonterminals. Such an SLP generates a single string.
Furthermore, we refer to a \emph{run-length straight line program} (\textsc{RLSLP}) as an
\textsc{SLP} that, in addition, allows rules of the form $A \rightarrow A_1^s$,
where $A, A_1$ are nonterminals and $s \in \mathbb{Z}_{\geq 2}$, which means that the right-hand side of the rule defining $A$ can be obtained by concatenating $s$ copies of $A_1$.

A \emph{parsing} is a way to decompose a string $S$ into non-overlapping \emph{blocks}, $S = S_1 \cdot S_2 \cdots S_k$.
A \emph{locally consistent parsing (LCP)}~\cite{locally_1} is a parsing where, if two positions $i,i'$
have long enough matching contexts $S[i-\alpha\dd i+\beta]=S[i'-\alpha\dd i'+\beta]$ and there is a block boundary
after $S[i]$, then there is also one after $S[i']$. The meaning of ``long enough'' depends on the LCP type~\cite{locally_1,coin_tossing,talg}.

\section{A New \texorpdfstring{$\delta$}{delta}-bounded RLSLP}\label{sec:grammar}
The measure $\delta$ was implicitly used in a stringology context already in~\cite{sublinear}, but it was formally defined later~\cite{talg} (as a way to construct an RLSLP of size $O(\gammalog)$ without knowing $\gamma$) and thoroughly studied in~\cite{delta}.
For a given string~$S$ and integer $k\ge 0$, let $d_k(S)$ be the number of distinct length-$k$ substrings in $S$.
 The sequence of all values $d_k(S)$ is known as the \emph{substring complexity} of $S$.
Then, $\delta$ is defined~as
\[\delta ~=~ \max\left\{\tfrac{d_k(S)}k: k\in \mathbb{Z}_{\ge 1}\right\}.\]
The bounds of~\cite{delta}, including $\gamma=\Oh(\deltalos)$, implicitly utilize the~following:

\begin{fact}\label{fct:2rSdeltalos}
  Every string $S\in [0\dd \sigma)^n$ of measure $\delta$ satisfies
\[\sum_{p=0}^\infty \tfrac{d_{2^p}(S)}{2^p} \le 5\delta + \deltalos = \Oh(\deltalos).\]
\end{fact}
\begin{proof}
  Let $\mu := \floor{\log\frac{\log \delta}{\log \sigma}}$
  and $\nu := \ceil{\log \frac{n}{\delta}}$.
  For $p\in [0\dd \mu]$, we use a bound  $d_{2^p}(S) \le \sigma^{2^p} \le \sigma^{2^\mu}$
  to derive \[\sum_{p=0}^\mu \frac{d_{2^p}(S)}{2^p}\le \sum_{p=0}^\mu \frac{\sigma^{2^\mu}}{2^p} \le 2 \cdot \sigma^{2^\mu}
  \le 2\cdot \sigma^{\frac{\log \delta}{\log \sigma}} = 2\delta.\]
  For $p\in (\mu \dd \nu)$,  we use a bound $d_{2^p}(S) \le \delta \cdot 2^p$
  to derive
  \begin{align*}\sum_{p=\mu+1}^{\nu-1} \frac{d_{2^p}(S)}{2^p}\le \delta\cdot(\nu-\mu-1) \le \delta\cdot \left(1+\log \tfrac{n}{\delta}-\log\tfrac{\log \delta}{\log \sigma}\right) &= \delta + \delta\log\tfrac{n\log\sigma}{\delta\log\delta}.\end{align*}
  For $p\ge \nu$, we use a bound $d_{2^p}(S) \le n$
  to derive 
  \[\sum_{p=\nu}^\infty \frac{d_{2^p}(S)}{2^p} \le \sum_{p=\nu}^\infty \frac{n}{2^p} 
  \le \frac{2n}{2^{\nu}} \le 2\delta.\]
  Overall, we get $\sum_{p=0}^\infty \tfrac{d_{2^p}(S)}{2^p}\le 2\delta+\delta+\delta\log\tfrac{n\log\sigma}{\delta\log\delta} + 2\delta=\Oh(\deltalos),$
where we note that $\delta\log\tfrac{n\log\sigma}{\delta\log\delta} = \Oh(\deltalos)$.\footnote{If $\delta\log n \le \sqrt{n}$, then $\delta\log\tfrac{n\log\sigma}{\delta\log\delta} \le \delta\log(n\log\sigma) \le 2\delta\log(\sqrt{n} \log \sigma)
\le 2 \deltalos$. Otherwise, $\log\delta > \frac{1}{2}\log n-\log \log n > \frac14 \log n$, so $\delta\log\tfrac{n\log\sigma}{\delta\log\delta} \le \delta\log\tfrac{4n\log\sigma}{\delta\log n}=\deltalos + 2\delta$.}
\end{proof}

An RLSLP of size $O(\deltalos)$ was built~\cite{delta} on top of the recompression method~\cite{jez}. 
In this section, we show that the same can be achieved on top of the block-based LCP~\cite{MSU97}.  
Unlike the previous construction, ours produces an RLSLP with $\Oh(\deltalos)$ rules in $O(n)$ deterministic time, though we still use randomization to ensure that the total grammar size is also $\Oh(\deltalos)$.

We adapt the preceding construction~\cite{delta}, which uses the so-called \emph{restricted recompression}~\cite{internal_pattern}. This technique pauses the processing for symbols whose expansion is too long for the current stage. 
A similar idea was used~\cite{locally,DSA} for adapting another LCP, called \emph{signature parsing}~\cite{parallel}. 
We apply restriction (the pausing technique) to the LCP of~\cite{MSU97} that forms blocks ending at local minima with respect to a randomized ordering of the alphabet.
This LCP will be used later to obtain near-optimal search time, extending previous work~\cite{talg}.
We call our parsing \emph{restricted block compression}.

\subsection{Restricted block compression}

Given a string $S \in \Sigma^+$, our restricted block compression builds a sequence
of strings $(S_k)_{k\ge 0}$ over the alphabet $\A$ defined recursively to contain symbols in~$\Sigma$, pairs formed by a symbol in $\A$ and an integer $m \ge 2$, and sequences of at least two symbols in $\A$; formally, $\A$ is the least
fixed point of the expression:
\[\A ~=~ \Sigma \cup (\A \times \mathbb{Z}_{\geq 2}) \cup  \bigcup_{i=2}^{\infty} \A^i \textrm{.}\]
In the following, we denote $\bigcup_{i=2}^{\infty} \A^i$ with $\A^{\geq 2}$.

Symbols in $\A\setminus \Sigma $ are \emph{nonterminals}, which are naturally associated with \emph{productions} $(A_1,\ldots, A_s) \rightarrow A_1\cdots A_s$ for $(A_1,\ldots, A_s) \in \A^{\geq 2}$ and $(A_1, s) \rightarrow A_1^s$ for $(A_1, s) \in \A \times \mathbb{Z}_{\geq 2}$.
Setting any $A \in \A$ as the starting symbol yields an RLSLP.\@
The string generated by this \textsc{RLSLP} is $\exp{A}$, where  $\mathtt{exp} : 
\A \rightarrow \Sigma^+$ is the \emph{expansion} function defined recursively:
\[ \exp{A} =
\begin{cases}
  A & \textrm{ if } A \in \Sigma, \\
  \exp{A_1}\cdots \exp{A_{s}} & \textrm{ if } A =
  (A_1,\ldots, A_{s}) \textrm{ for } A_1,\ldots,A_{s} \in \A, \\
  \exp{A_1}^s & \textrm{ if } A = (A_1, s) \textrm{ for } A_1 \in \A
  \textrm{ and } s \in \mathbb{Z}_{\geq 2}.
\end{cases}
\]
The expansion function is extended homomorphically to $\mathtt{exp} : \A^* \rightarrow
\Sigma^*$, with $\exp{A_1\cdots A_m} = \exp{A_1}\cdots\exp{A_m}$ for $A_1\cdots A_m \in \A^*$.

Starting from $S_0 = S$, the strings $(S_k)_{k=0}^\infty$ such that $\exp{S_k}=S$ are built by the
alternate applications of two functions, both of which decompose a
string $T \in \A^+$ into \emph{blocks} (by placing \emph{block boundaries} between some characters) and then collapse 
blocks of length $s\ge 2$ into individual symbols in~$\A$.
In \cref{def:run-length}, the blocks are maximal \emph{runs} of the same symbol in a subset $\B \subseteq \A$, and they are
collapsed to symbols in $\A \times \mathbb{Z}_{\geq 2}$. 

\begin{definition}[Run-length encoding]\label{def:run-length}
Given $T \in \A^+$ and a subset of symbols $\B \subseteq \A$,
we define
$rle_{\B}(T) \in \mathcal{A^+}$ as the string obtained by decomposing $T$ into blocks and collapsing these blocks as follows:
\begin{enumerate}
\item{For every $i \in [1 \dd \vert T\vert )$, place a block boundary between $T[i]$ and $T[i +
1]$ if $T[i]\notin \B$, $T[i+1]\notin \B$, or $T[i]\ne T[i+1]$.}
\item{For every block $T[i\dd i+s)$ of $s\ge 2$ equal symbols $A$, replace $T[i \dd {i + s}) = A^s$ with the symbol $(A, s) \in \A$.}
\end{enumerate}
\end{definition}

In \cref{def:block}, the block boundaries are determined by the local minima of a function of 
the set $\Sigma(T)\subseteq \A$ of symbols that occur in $T$, and the blocks are collapsed to symbols in $\A^{\geq 2}$.

\begin{definition}[Local minima]\label{def:minima}
  Given $T \in \A^+$, we say that a position $i \in (1 \dd \vert T\vert )$ is a \emph{local minimum} with respect to  a function $\pi: \Sigma(T) \rightarrow \mathbb{Z}$ if
  \[\pi(T[i - 1]) > \pi(T[i])  \text{ and } \pi(T[i]) < \pi(T[i + 1]).\]
\end{definition}

\begin{definition}[Restricted block parsing]\label{def:block}
Given $T \in \A^+$, a function $\pi : \Sigma(T) \rightarrow \mathbb{Z}$, and a subset of symbols $\B \subseteq \A$, we define
$bc_{\pi,\B}(T) \in \mathcal{A^+}$ as the string obtained by decomposing $T$ into blocks and collapsing these blocks as follows:
\begin{enumerate}
\item{For every $i \in [1 \dd \vert T\vert )$, place a block boundary between $T[i]$ and $T[i +
1]$ if $T[i]\notin \B$, $T[i+1]\notin \B$, or $i$ is a local minimum with respect to $\pi$.}
\item{For each block $T[i \dd i + s)$ of length $s\ge 2$, replace $T[i \dd i + s)$ with a symbol $(T[i],\ldots,T[i+s-1]) \in \A$.}
\end{enumerate}
\end{definition}

Note that $\B$ consists of \emph{active} symbols that can be combined into larger blocks; we say that the other symbols are \emph{paused}.
The idea of our restricted block compression is to create successive strings $S_k$, starting from $S_0=S$. At the odd levels $k$, we perform run-length encoding on the preceding string~$S_{k-1}$. On the even levels $k$, we perform block parsing on the preceding string~$S_{k-1}$.  
We pause the symbols whose expansions are too long for that level.

\begin{definition}[Restricted block compression]\label{construction}
Given $S \in \Sigma^+$, the strings $S_k$ for $k \in \Zz$ are constructed as follows, 
where $\ell_k := \big(\frac{4}{3}\big)^{\lceil {k/2} \rceil - 1}$, $\A_k := \{A \in \A : \vert \exp{A} \vert \leq \ell_k\}$,
 and $\pi_k : \Sigma(S_{k-1}) \rightarrow [1\dd \vert \Sigma(S_{k-1})\vert]$ is a bijection satisfying $\pi_k(A) < \pi_k(B)$ 
 for $A \in \Sigma(S_{k-1}) \setminus \A_k$ and $B\in \Sigma(S_{k-1}) \cap \A_k$:
\begin{itemize}
\item{If $k=0$, then $S_k = S$.}
  \item{If $k > 0$ is odd, then $S_k = rle_{\A_k}(S_{k-1})$.}
  \item{If $k > 0$ is even, then $S_k = bc_{\pi_k, \A_k}(S_{k-1})$.}
\end{itemize}
\end{definition}

\subsection{Grammar size analysis}

Our RLSLP will be built by performing restricted block compression as long as $\vert S_k\vert >1$. Although the resulting \textsc{RLSLP} formally has infinitely many symbols, we can remove those having no occurrences in any $S_k$. To define the actual symbols in the grammar, for all $k \in \Zz$, denote $\S_k := \{ S_k[j] : j \in [1 \dd \vert S_k\vert ] \}$ and $\S := \bigcup_{k = 0}^\infty \S_k$. 

Recall that $\exp{S_k}=S$ holds for every $k\in \Zz$. Based on this, we associate $S_k$ with a decomposition of $S$ into \emph{phrases}.

\begin{definition}[Phrase boundaries]
  For every $k \in \Zz$ and $j \in [1 \dd \vert S_k\vert ]$, we define the level-$k$ \emph{phrases} of $S$ induced by $S_k$ as the fragments
  \[S(\vert \exp{S_k[1 \dd j)}\vert  \dd \vert \exp{S_k[1 \dd j]}\vert ] =
  \exp{S_k[j]}.\]
  We also define the set $B_k\subseteq [0\dd n]$ of \emph{phrase  boundaries} induced by $S_k$:
\[ B_k = \{ \vert \exp{S_k[1\dd j]} \vert : j \in [0\dd \vert S_k\vert ]\}. \]
Note that taking $j=0$ and $j=\vert S_k \vert$ yields $0\in B_k$ and $n\in B_k$, respectively.
\end{definition}

The following lemma captures the local consistency of restricted block compression: the phrase boundaries are determined by a small context.

\begin{lemma}\label{lemma:alphalk}
Consider integers $k,m,\alpha \in \Zz$ with $\alpha \ge \floor{8\ell_k}$
and $i,i'\in [m+2\alpha \dd n - \alpha]$ such that $S(i-m-2\alpha \dd i + \alpha] = S(i' -m- 2\alpha \dd i' + \alpha]$.
\begin{enumerate}
  \item If $i\in B_k$, then $i'\in B_k$.
  \item If $S(i-m\dd i]$ is a level-$k$ phrase, then $S(i'-m\dd i']$ is a level-$k$ phrase corresponding to the same symbol in $S_k$.
\end{enumerate}
\end{lemma} 
\begin{proof}
We proceed by induction on $k$, with a weaker assumption $\alpha \geq \floor{7\ell_k}$ for odd~$k$. 
In the base case of $k=0$, the claim is trivial because $B_0 =[0 \dd n]$ and every character of $S$ forms a level-$0$ phrase. 
Next, we prove that the claim holds for integers $k > 0$ and $\alpha > \floor{\ell_k}$ assuming that it holds for $k-1$ and $\alpha - \floor{\ell_k}$. 
This is sufficient for the inductive step: 
If $\alpha \geq \floor{8\ell_k}$ for even $k>0$, then $\alpha - \floor{\ell_k} \geq \floor{7\ell_k} = \floor{7\ell_{k-1}}$.
Similarly, if $\alpha \geq \floor{7\ell_k}$ for odd $k$, then $\alpha - \floor{\ell_k} \geq \floor{6\ell_k} = \floor{8\ell_{k-1}}$.

We start with the first item, where we can assume $m=0$ without loss of generality.
For a proof by contradiction, suppose that $S(i - 2\alpha \dd i + \alpha] = S(i' - 2\alpha \dd i' + \alpha]$ and $i \in B_k$ yet $i' \notin B_k$.
By the first item of the inductive assumption applied to positions $i$ and $i'$, we conclude that $i \in B_k \subseteq B_{k-1}$ implies $i' \in B_{k-1}$. 
Let us set $j, j'\in [1\dd \vert S_{k-1}\vert )$ so that $i = \vert \exp{S_{k-1}[1 \dd j]}\vert $ and $i' = \vert \exp{S_{k-1}[1 \dd j']}\vert$. 
Since $i\in B_k$ yet $i'\notin B_k$, the parsing of $S_{k-1}$ places a block boundary
between $S_{k-1}[j]$ and $S_{k-1}[j+1]$, but it does not place a block boundary 
between $S_{k-1}[j']$ and $S_{k-1}[j'+1]$.
By Definitions~\ref{def:run-length} and~\ref{def:block}, the latter yields $S_{k-1}[j'], S_{k-1}[j'+1] \in \A_k$.
Consequently, the level-$(k-1)$ phrases $S(i' - \ell \dd i'] := \exp{S_{k-1}[j']}$ and $S(i' \dd i' + r] := \exp{S_{k-1}[j'+1]}$ around position
$i'$ are of length at most $\floor{\ell_k}$ each.
Therefore, the assumption $S(i-2\alpha\dd i+\alpha]=S(i'-2\alpha\dd i'+\alpha]$
implies $S(i'-\ell-2(\alpha-\floor{\ell_k})\dd i'+(\alpha-\floor{\ell_k})]=
S(i-\ell-2(\alpha-\floor{\ell_k})\dd i+(\alpha-\floor{\ell_k})]$
as well as $S((i'+r)-r-2(\alpha-\floor{\ell_k})\dd i'+r+(\alpha-\floor{\ell_k})]=
S((i+r)-r-2(\alpha-\floor{\ell_k})\dd i+r+(\alpha-\floor{\ell_k})]$.
Thus, we can apply (the second item of) the inductive assumption for the level-$(k-1)$ phrases $S(i'-\ell\dd i']$ and $S(i'\dd i'+r]$, concluding that $S(i-\ell \dd i]$ and $S(i \dd i+ r]$ are also level-$(k-1)$ phrases parsed into $S_{k-1}[j]=S_{k-1}[j']$ and $S_{k-1}[j+1]=S_{k-1}[j'+1]$,
respectively.

If $k$ is odd, then a boundary between two symbols in $\A_{k}$ is placed if and only if the two symbols differ.
Consequently, $S_{k-1}[j']=S_{k-1}[j'+1]$ and $S_{k-1}[j]\ne S_{k-1}[j+1]$. This contradicts $S_{k-1}[j]=S_{k-1}[j']$ and $S_{k-1}[j+1]=S_{k-1}[j'+1]$.

Thus, it remains to consider the case of even $k$.
Since the block parsing placed a boundary between $S_{k-1}[j],S_{k-1}[j+1]\in \A_{k}$, we conclude from \cref{def:block} that $j$ must be a local minimum with respect to $\pi_k$, i.e., $\pi_k(S_{k-1}[j-1]) > \pi_k(S_{k-1}[j]) < \pi_k(S_{k-1}[j+1])$.
Due to $S_{k-1}[j]\in \A_k$, the condition on $\pi_k$ imposed in \cref{construction} implies $S_{k-1}[j-1]\in \A_k$.
Consequently, the level-$(k-1)$ phrase $S({i-\ell-\ell'}\dd {i-\ell}] := \exp{S_{k-1}[{j-1}]}$ is of length at most $\floor{\ell_k}$, and thus $\ell+\ell'\le 2\floor{\ell_k}$.
Therefore, the assumption $S(i-2\alpha\dd i+\alpha]=S(i'-2\alpha\dd i'+\alpha]$
implies $S((i-\ell)-\ell'-2(\alpha-\floor{\ell_k})\dd (i-\ell)+(\alpha-\floor{\ell_k})]=
S((i'-\ell)-\ell'-2(\alpha-\floor{\ell_k})\dd (i'-\ell)+(\alpha-\floor{\ell_k})]$.
Hence, we can apply (the second item of) the inductive assumption for the level-$(k-1)$ phrase $S(i-\ell-\ell'\dd i-\ell]$, concluding that $S(i'-\ell-\ell' \dd i'-\ell]$ is a level-$(k-1)$ phrase parsed into $S_{k-1}[j'-1]=S_{k-1}[j-1]$.
Thus, $\pi_k(S_{k-1}[j'-1])=\pi_k(S_{k-1}[j-1]) > \pi_k(S_{k-1}[j']) = \pi_k(S_{k-1}[j]) < \pi_k(S_{k-1}[{j'+1}]) = \pi_k(S_{k-1}[{j+1}])$,
which means that $j'$ is a local minimum with respect to $\pi_k$ and contradicts $i'\notin B_k$ by~{\cref{def:block}}.

Let us proceed to the proof of the second item.
Let $S_{k-1}(j-m'\dd j]$ be the block corresponding to the level-$k$ phrase $S(i-m\dd i]$.
By the inductive assumption applied separately to every level-$(k-1)$ phrase contained in $S(i-m\dd i]$, the fragment $S({i'-m}\dd i']$ consists of level-$(k-1)$ phrases that, in $S_{k-1}$,
are collapsed into a fragment $S_{k-1}(j'-m'\dd j']$ matching $S_{k-1}(j-m'\dd j]$.
Moreover, by the already proved first item, applied to corresponding positions within $[i-m\dd i]$ and $[i'-m\dd i']$, the parsing of $S_{k-1}$ places block boundaries around $S_{k-1}(j'-m'..j']$ but nowhere within that fragment.
Hence, $S_{k-1}(j-m'\dd j]$ and $S_{k-1}(j'-m'\dd j']$ are matching blocks,
which means that they are collapsed into matching symbols of $S_k$.
The phrases $S(i-m\dd i]$ and $S(i'-m\dd i']$ are thus represented by matching symbols in $S_k$.
\end{proof}

Our next goal is to prove that the phrase boundaries are locally sparse, that is, their number within any interval $I$ decreases exponentially with $k$. This holds because two out of any three consecutive blocks must be merged in the even levels, unless one is already long for that level (and thus paused).

\begin{lemma}\label{lemma:BkI}
  For every $k \in \Zz$ and non-empty interval $I \subseteq [0 \dd n]$,
  we have
  \[ \vert B_k \cap I\vert < 2 + \tfrac{4\vert I \vert-4}{\ell_{k+1}} \textrm{.}\]
  \end{lemma}
  
  \begin{proof}
  We proceed by induction on $k$. For $k = 0$, we have $\vert B_k \cap I\vert = \vert I \vert < 2 +   4\vert I \vert -4 = 2 + \frac{4\vert I \vert-4}{\ell_{1}}$. If $k$ is odd, we note that $\ell_{k+1}=\ell_k$ and $B_k \subseteq  B_{k-1}$.
  Consequently, \[\vert B_k \cap I\vert  \leq \vert B_{k-1} \cap I\vert < 2 +
  \tfrac{4\vert I\vert -4}{\ell_{k}} = 2 + \tfrac{4\vert I \vert-4}{\ell_{k+1}}.\]
  If $k$ is even, let us  define
  \begin{equation*}
    \begin{aligned}
      J &= \{j \in [1\dd \vert S_{k-1}\vert) : \vert \exp{S_{k-1}[1\dd j-1]}\vert \in I\text{ and }\vert \exp{S_{k-1}[1\dd j+1]}\vert \in I\},\\
      J' &= \{ j \in J: S_{k-1}[j] \in \A_k\text{ and }S_{k-1}[j+1] \in \A_k\},\\
      \end{aligned}
    \end{equation*}
  Observe that $\vert J \vert = \max(0, \vert B_{k-1}\cap I\vert-2)$: every $j\in J$ corresponds to a position $\vert \exp{S_{k-1}[1\dd j]} \vert\in B_{k-1}\cap I$ which is neither the leftmost nor the rightmost one in $B_{k-1}\cap I$.
  If $j\notin J'$, then $S_{k-1}[j]\notin \A_k$ or $S_{k-1}[j+1]\notin \A_k$.
  The level-$(k-1)$ phrase $S(p\dd q]$ corresponding to this pausing symbol satisfies $[p\dd q]\subseteq I$
  and, by definition of $\A_k$, it is of length $q-p > \ell_k$. There are at most $\frac{\vert I \vert -1}{\ell_k}$ phrases $S(p\dd q]$
  satisfying both conditions, and thus $\vert J\setminus J' \vert \le 2\cdot \frac{\vert I \vert-1}{\ell_k}$.

  Every position $i\in B_k\cap I$, except for the leftmost and the rightmost one, corresponds to a position $j\in J$ such that $i=\vert \exp{S_{k-1}[1\dd j]}\vert$ and the restricted block parsing places a block boundary between $S_{k-1}[j]$ and $S_{k-1}[j+1]$.
  If $j\in J'$, then $j$ must be a local minimum with respect to $\pi_k$ and, in particular, $\pi_k(S_{k-1}[j-1])>\pi_k(S_{k-1}[j])$.
  By the condition on $\pi_k$ specified in \cref{construction}, this implies that $S_{k-1}[j-1]\in \A_k$. 
  Thus, the restricted block parsing does not place a block boundary between $S_{k-1}[j-1]$ and $S_{k-1}[j]$.
  Since $i$ is not the leftmost position in $B_k\cap I$, we must have $j-1\in J'$.
  Overall, using the inductive assumption, we conclude that 
  \begin{multline*}
  \vert B_k \cap I \vert \leq 2 + \vert J\setminus J' \vert + \tfrac12\vert J'\vert = 2 + \tfrac12 \vert J\setminus J' \vert + \tfrac12 \vert J\vert \\
  = 2 + \tfrac12 \vert J\setminus J' \vert  + \max\left(0, \tfrac12\vert B_{k-1} \cap I \vert-1\right) 
  < 2 +\tfrac{\vert I \vert-1}{\ell_k} + \tfrac{2\vert I \vert-2}{\ell_k}
  = 2 +  \tfrac{4\vert I \vert-4}{\ell_{k+1}}.\tag*{\qedhere}
  \end{multline*}
  \end{proof}

Plugging $I=[0\dd n]$, we derive an upper bound on $\vert S_k\vert=\vert B_k \vert-1$.
\begin{corollary}\label{cor:sizeSk}
For every $k \in \Zz$, we have $\vert S_k\vert  < 1 + \frac{4n}{\ell_{k+1}}$.
Consequently, $\vert S_{\kappa} \vert = 1$ holds for $\kappa:= 2\ceil{\log_{4/3}(4n)}$.
\end{corollary}

The following lemma essentially bounds the substring complexity of each string $S_k$ in terms of the substring complexity of $S$.
An important detail, however, is that we only consider substrings of $S_k$ consisting of active symbols (in $\A_{k+1}$).
This is sufficient because the pausing symbols form length-1 blocks that are copied to $S_{k+1}$ without getting collapsed.

\begin{lemma}\label{lemma:SkAk+1}
For all integers $k\ge 0$ and $m\ge 1$,
the number of distinct substrings of $S_k$ belonging to $\A_{k+1}^m$ is $\Oh(m(1+\frac{d_{r}(S)}{r}))$, where $r = 2^{\lceil \log(50m\ell_{k+1}) \rceil}$.
\end{lemma}
\begin{proof}
Denote $\alpha := \floor{8\ell_{k}}$ and $\ell := \frac{r}{2}\ge 25m\ell_{k+1} \ge 3\alpha+m\floor{\ell_{k+1}}$,
  and let $L$ be the set consisting of the trailing $\ell$ positions in $S$ and all positions covered by the leftmost occurrences of substrings of $S$ of length at most $\ell$. We first prove two auxiliary claims.
  
  \begin{claim}\label{claim:POdelta}
    The string $S_k$ contains at most $\vert L\cap B_k\vert$ distinct substrings in $\A_{k+1}^m$.
  \end{claim}
  \begin{proof}
  Let us fix a substring $T\in \A_{k+1}^m$ of $S_k$ and let $S_{k}(j-m\dd j]$ be the leftmost occurrence of $T$ in $S_k$.
  Moreover, let $p=\vert \exp{S_k[1 \dd j-m]}\vert $ and $q = \vert \exp{S_k[1 \dd j]}\vert $ so that $S(p\dd q]$ is the expansion of $S_{k}(j-m\dd j]$.
  By $S_{k}(j-m\dd j]\in \A_{k+1}^m$, we have $q-p \le m\lfloor{\ell_{k+1}}\rfloor \le \ell-3\alpha$.

  Suppose that $q\notin L$.
  Due to $(0\dd \ell]\cup (n-\ell\dd n]\subseteq L$, this implies that $q\in (\ell\dd n-\ell]$
  is not covered by the leftmost occurrence of any substring of length at most $\ell$.
  In particular, $S(p-2\alpha \dd q+\alpha]$ must have an earlier occurrence $S(p'-2\alpha\dd q'+\alpha]$
  for some $p'<p$ and $q'<q$.
  Consequently, \cref{lemma:alphalk}, applied to subsequent level-$k$ phrases comprising $S(p\dd q]$,
  shows that $S(p'\dd q']$ consists of full level-$k$ phrases and the corresponding fragment of $S_k$ matches $S_k(j-m\dd j]=T$. 
  By $q'<q$, this contradicts the assumption that $S_{k}(j-m\dd j]$ is the leftmost occurrence of $T$
  in~$S_k$. 
  
  Thus, we must have $q\in L$. 
  A level-$k$ phrase ends at position $q$, so we also have $q\in B_k$. Since the position $q$ uniquely determines the substring $T\in \A_{k+1}^m$,
  this yields an upper bound of $\vert L\cap B_k\vert $ on the number of choices for $T$.
  \end{proof}

\begin{claim}\label{clm:cover}
  The set $L$ forms $\Oh(1+\frac{d_{r}(S)}{r})$ intervals of total length $\Oh(r+d_{r}(S))$.
\end{claim}
\begin{proof}
  Each position in $L\cap (0\dd n-\ell]$ is covered by the leftmost occurrence of a length-$\ell$ substring of $S$, and thus $L$ forms at most 
  $\lfloor \frac{1}{\ell}\vert L\vert \rfloor$ intervals of length at least $\ell$ each. 
  Hence, it suffices to prove that the total length satisfies $\vert L\vert =O(r+d_{r}(S))$.
  For this, note that, for each position $j\in L\cap [\ell \dd n-\ell]$,
  the fragment $S(j-\ell\dd j+\ell]$ is the leftmost occurrence of a length-$r$ substring of $S$; this because any length-$\ell$ fragment covering position $j$ is contained within $S(j-\ell\dd j+\ell]$. Consequently, $\vert L\vert \le r+d_{r}(S)$ holds as claimed.
\end{proof}

By \cref{claim:POdelta}, it remains to bound $\vert L \cap B_k\vert$. 
Let $\I$ be the family of intervals covering $L$. For each $I \in \I$, \cref{lemma:BkI}
implies $\vert B_k \cap I\vert  \leq 1 + \frac{4\vert I\vert }{\ell_{k+1}}$. By the bounds on $\I$ following from \cref{clm:cover}, this yields the following result:
\[ \vert B_k \cap L\vert  \leq \vert \I\vert  + \tfrac{4}{\ell_{k+1}}\sum_{I \in
  \I}\vert I\vert  = O\left(1+\tfrac{d_{r}(S)}{r}+ \tfrac{r + d_{r}(S)}{\ell_{k+1}}\right) = O\left(m\left(1+\tfrac{d_{r}(S)}{r}\right)\right),\]
  where the last inequality holds because $r = \Theta(m\ell_{k+1})$.
  \end{proof}

The following result combines \cref{fct:2rSdeltalos,cor:sizeSk,lemma:SkAk+1}.

\begin{corollary}\label{corollary:deltalog}
For every string $S\in [0\dd \sigma)^n$ with measure $\delta$, we have $\vert \S\vert = O(\deltalos)$.
\end{corollary}
\begin{proof}
  Note that $\vert \S\vert  \leq 1 + \sum_{k=0}^{\kappa-1} \vert \S_k \setminus \S_{k+1}\vert $.
  First, we observe that \cref{construction} guarantees $\S_k
  \setminus \S_{k+1} \subseteq \S_k \cap \A_{k+1}$. Moreover, each symbol in $\S_k\cap \A_{k+1}$
  corresponds to a distinct character of $S_{k}$ contained in $\A_{k+1}$,
  and thus $\vert \S_k \setminus \S_{k+1}\vert  \leq \vert \S_k \cap \A_{k+1}\vert  = \Oh\big(1+\frac{d_{r_k}(S)}{r_k}\big)$ holds due to 
  \cref{lemma:SkAk+1} for $r_k := 2^{\lceil \log (50\ell_{k+1}) \rceil}$.
  Observe that, for every $p\in [0\dd \ceil{\log(50\ell_{\kappa})}]$, there are at most 6 integers $k$ such that $r_k = 2^p$
  (this is because $\ell_{k+6} = (4/3)^3 \ell_{k} > 2\ell_k$).
  Hence, Fact~\ref{fct:2rSdeltalos} and the definition of $\kappa$ yield
  \[\sum_{k=0}^{\kappa-1} \vert \S_k \cap \A_{k+1}\vert =\!\!\!\!\!\!\sum_{p=0}^{\ceil{\log(50\ell_{\kappa})}}\!\!\!\!\!\! \Oh\left(1+\tfrac{d_{2^p}(S)}{2^p}\right) = \Oh(\log n + \deltalos) = \Oh(\deltalos).\]
Overall, $\vert \S\vert  = 1 + O(\deltalos) = O(\deltalos)$ holds as claimed.
\end{proof}

Next, we show that the total expected grammar size is $\Oh(\deltalos)$.

\begin{theorem}\label{thm:grl}
Consider the restricted block compression of a string $S\in [0\dd \sigma)^n$ with measure $\delta$, 
where the functions $(\pi_k)_{k\ge 0}$ in \cref{construction} are chosen uniformly at random.
Then, the expected size of the resulting RLSLP is $\Oh(\deltalos)$.
\end{theorem}
\begin{proof}
Although Corollary~\ref{corollary:deltalog} guarantees that $\vert \S\vert =\Oh(\deltalos)$, the remaining problem is that the size of the resulting grammar (i.e., the sum of production sizes) can be larger. 
  Every symbol in $\Sigma \cup (\A\times \mathbb{Z}_{\ge 2})$ contributes $\Oh(1)$ to the RLSLP size,
so it remains to bound the total size of productions corresponding to symbols in $\A^{\ge 2}$.
These symbols are introduced by restricted block parsing, i.e., they belong to $\S_{k+1}\setminus \S_k$ for odd $k \in [0\dd \kappa)$.
In order to estimate their contribution to the grammar size, we shall fix $\pi_0,\ldots,\pi_k$ and compute the expectation with respect to the random choice of $\pi_{k+1}$. In this setting, we prove the following claim:

\begin{claim}\label{clm:block}
  Let $k > 0 $ be odd and $T\in \A_{k}^m$ be a substring of $S_{k}$. Restricted block parsing $bc_{\pi_{k+1}, \A_{k+1}}(S_{k})$ creates a block matching $T$
  with probability $\Oh(2^{-m})$. 
  \end{claim}
  \begin{proof}
  Since $S_k=rle_{\A_{k}}(S_{k-1})$ and $\A_{k+1}=\A_k$, every two subsequent symbols of $T$ are distinct.
  Observe that if $T$ forms a block, then there is a value $t\in [1\dd m]$
  such that $\pi_{k+1}(T[1])<\cdots  < \pi_{k+1}(T[t]) > \cdots > \pi_{k+1}(T[m])$;
  otherwise, there would be a local minimum within every occurrence of $T$ in $S_{k-1}$.
  In particular, denoting $h:=\lfloor{m/2}\rfloor$,  we must have $\pi_{k+1}(T[1])<\cdots  < \pi_{k+1}(T[h+1])$ (when $t > h$)
  or $\pi_{k+1}(T[m-h])>\cdots  > \pi_{k+1}(T[m])$ (when $t \le h$).
  However, the probability that the values $\pi_{k+1}(\cdot)$ for $h+1$ consecutive characters form a strictly increasing (or strictly decreasing)
  sequence is at most $\frac{1}{(h+1)!}$: either exactly $\frac{1}{(h+1)!}$ (if the characters are distinct) or $0$ (otherwise);
  this is because $\pi_{k+1}$ shuffles $\Sigma(S_k)\cap \A_{k+1}$ uniformly at random.
  Overall, we conclude that the probability that $T$ forms a block does not exceed $\frac{2}{(h+1)!}\le 2^{-\Omega(m\log m)} \le  \Oh(2^{-m})$.
\end{proof}

Next, note that every symbol in $\S_{k+1}\setminus \S_k$ is obtained by collapsing a block of $m$ active symbols
created within $bc_{\pi_{k+1}, \A_{k+1}}(S_{k})$ (with distinct symbols obtained from distinct blocks).
By \cref{lemma:SkAk+1}, the string $S_k$ has  $\Oh\left(m\left(1+\frac{d_{r_{k,m}}(S)}{r_{k,m}}\right)\right)$ distinct substrings $T\in \A_{k+1}^m$
for $r_{k,m}:= 2^{\lceil\log (50m\ell_{k+1})\rceil}$.
By \cref{clm:block}, any fixed substring $T\in \A_{k+1}^m$ yields a symbol in $\S_{k+1}\setminus \S_k$ with probability $\Oh(2^{-m})$.
Consequently, the total contribution of symbols in $\S_{k+1}\setminus \S_k$ to the RLSLP size is, in expectation, 
\[\sum_{m=2}^\infty \Oh\left(\frac{m}{2^m}\left(1+\frac{d_{r_{k,m}}(S)}{r_{k,m}}\right)\right)=\Oh(1)+\sum_{m=2}^\infty \Oh\left(\frac{m}{2^m}\cdot \frac{d_{r_{k,m}}(S)}{r_{k,m}}\right).\]

Across all the odd levels $k\in [0\dd \kappa)$, this sums up to at most
\begin{multline*}\Oh(\kappa)+\sum_{k=0}^\infty \sum_{m=2}^\infty \Oh\left(\frac{m}{2^m}\cdot \frac{d_{r_{k,m}}(S)}{r_{k,m}}\right)\\
= \Oh(\log n) + \sum_{p=0}^\infty \sum_{m=2}^\infty \sum_{\substack{k\in \Zz\\ r_{k,m}=2^p}} \Oh\left(\frac{m}{2^m}\cdot \frac{d_{2^p}(S)}{2^p}\right).
\end{multline*}
For every fixed $p,m\in \Zz$, the number of integers $k\in \Zz$ satisfying $r_{k,m}=2^p$ is at most 6
(because $\ell_{k+6}=(4/3)^3 \ell_{k} > 2\ell_k$). Hence, for every $p\in \Zz$,
we have \[\sum_{m=2}^\infty \sum_{\substack{k\in \Zz\\ r_{k,m}=2^p}} \Oh\left(\frac{m}{2^m}\cdot \frac{d_{2^p}(S)}{2^p}\right)=
\sum_{m=2}^\infty \Oh\left(\frac{m}{2^m}\cdot \frac{d_{2^p}(S)}{2^p}\right)=\Oh\left(\frac{d_{2^p}(S)}{2^p}\right).\]
Consequently, Fact~\ref{fct:2rSdeltalos} implies that the total expected grammar size is
\[\Oh(\log n) + \sum_{p=0}^\infty \Oh\left(\frac{d_{2^p}(S)}{2^p}\right) = \Oh(\log n) + \Oh(\deltalos) = \Oh(\deltalos).\tag*{\qedhere}\]
\end{proof}

We are now ready to show how to build an RLSLP of size $\Oh(\deltalos)$ in linear expected time.

\begin{corollary}\label{cor:grl}
Given a string $S\in [0\dd \sigma)^n$ with measure $\delta$, we can build an RLSLP of size $\Oh(\deltalos)$ in $O(n)$ expected time.
\end{corollary}
\begin{proof}
We apply \cref{construction} on top of the given string $S$, with functions $\pi_k$ chosen uniformly at random.
It is an easy exercise to carry out this construction in $O(\sum_{k=0}^{\kappa} \vert S_k\vert ) = O(n)$
worst-case time.

The expected size of the resulting RLSLP is at most $c \cdot \deltalos$ for some constant $c$;
we can repeat the construction (with fresh randomness) until it yields an RLSLP of size at most $2c \cdot \deltalos$.
By Markov's inequality, we succeed after $O(1)$ attempts in expectation. As a result, in $O(n)$ expected time, we obtain a grammar of total worst-case size $\Oh(\deltalos)$.
\end{proof}

\begin{remark}[Grammar height]\label{observation:ndelta}
In the algorithm of Corollary~\ref{cor:grl}, we can terminate restricted block compression after
$\lambda := 2\lfloor {\log_{4/3}\frac{n}{\delta}} \rfloor$ levels and complete the grammar with an initial symbol rule 
$A_{\lambda} \rightarrow S_{\lambda}[1]\cdots S_{\lambda}[\vert S_{\lambda}\vert ]$ so that $\exp{A_\lambda}=S$. 
Corollary~\ref{cor:sizeSk} yields $\vert S_{\lambda}\vert  = O(1+(\frac{3}{4})^{\lambda/2}n) = O(\delta)$,
so the resulting RLSLP is still of size $\Oh(\deltalos)$; however, the height is now $O(\log\frac{n}{\delta})$.
\end{remark}

\section{Local Consistency Properties}

We now show that the local consistency properties of our grammar enable fast indexed searches. Previous work~\cite{talg} achieves this by showing that, thanks to the locally consistent parsing, only a set $M(P)$ of $O(\log \vert P\vert )$ pattern positions need to be analyzed for searching. To use this result, we now must take into account the pausing of symbols. Surprisingly, this modification allows for a much simpler definition of $M(P)$.

\begin{definition}\label{def:M}
  For every non-empty fragment $S[i\dd j]$ of $S$, we define
  \[B_{k}(i, j) = \{ p - i : p \in B_k\cap [i\dd j) \}\]
  and 
  \begin{multline*}
  M(i,j) = \bigcup_{k \ge 0} \big(B_k(i,j)\setminus [2\alpha_{k+1}\dd j-i-\alpha_{k+1}) \\
 \cup \{\min (B_k(i,j)\cap [2\alpha_{k+1}\dd j-i-\alpha_{k+1}))\}\big),
  \end{multline*}
  where $\alpha_k = \floor{8\ell_k}$ and $\{\min \emptyset\} = \emptyset$.
\end{definition}

Intuitively, the set $B_k(i,j)$ lists (the relative locations of) all level-$k$ phrase boundaries inside $S[i\dd j]$.
For each level $k\ge 0$, we include in $M(i,j)$ the phrase boundaries that are close to either of the two endpoints of $S[i\dd j]$
(in the light of \cref{lemma:alphalk}, it may depend on the context of $S[i\dd j]$ which of these phrase boundaries
are preserved as level-$(k+1)$ boundaries) as well as the leftmost phrase boundary within the remaining internal part of $S[i\dd j]$.

\begin{lemma}\label{lem:m}
The set $M(i,j)$ satisfies the following properties:
\begin{enumerate}
\item\label{it:min} For each $k\ge 0$, if $B_k(i,j)\ne \emptyset$, then $\min B_{k}(i,j) \in M(i,j)$.
\item\label{it:size} We have $\vert M(i,j)\vert =\Oh(\log (j-i+2))$.
\item\label{it:cons} If $S[i'\dd j']=S[i\dd j]$, then $M(i',j')=M(i,j)$.
\end{enumerate}
\end{lemma}
\begin{proof}
Let us express $M(i,j)=\bigcup_{k\ge 0} M_k(i,j)$, setting
\[M_k(i,j) := B_k(i,j)\setminus [2\alpha_{k+1}\dd j-i-\alpha_{k+1}) \cup \{\min(B_k(i,j)\cap [2\alpha_{k+1}\dd j-i-\alpha_{k+1}) )\}.\]

As for \cref{it:min}, it is easy to see that $\min B_{k}(i,j)\in M_k(i,j)$: we consider two cases, depending on whether $\min B_{k}(i,j)$ belongs to $[2\alpha_{k+1}\dd j-i-\alpha_{k+1})$ or not.

As for \cref{it:size}, let us first argue that $\vert M_k(i,j)\vert =\Oh(1)$ holds for every $k\ge 0$.
Indeed, each element $q\in B_k(i,j)\cap [0\dd 2\alpha_{k+1})$ corresponds to $q+i\in B_k\cap [i\dd i+2\alpha_{k+1})$
and each element $q\in B_{k}(i,j)\cap [j-i-\alpha_{k+1}\dd j-i)$ corresponds to $q+i\in B_k\cap [j-\alpha_{k+1}\dd j)$.
By \cref{lemma:BkI}, we conclude that $\vert M_k(i,j)\vert \le 1+(2+\frac{8\alpha_{k+1}}{\ell_{k+1}})+(2+\frac{4\alpha_{k+1}}{\ell_{k+1}})=\Oh(1)$.
Moreover, if $\ell_{k} > 4(j-i)$, then \cref{lemma:BkI} further yields $\vert B_k(i,j)\vert =\vert B_k\cap [i\dd j)\vert \le 1$.
Since $M_k(i,j)$ and $B_{k+1}(i,j)$ are both subsets of $B_k(i,j)$, this means that $\big\vert \bigcup_{k : \ell_k > 4(j-i)} M_k(i,j)\big\vert \le 1$.
The number of levels $k$ satisfying $\ell_k \le 4(j-i)$ is $\Oh(\log(j-i+2))$,
and thus \[\vert M(i,j)\vert \le \Oh(1)\cdot \Oh(\log(j-i+2))+1=\Oh(\log (j-i+2)).\]

As for \cref{it:cons}, we shall prove by induction on $k$ that $M_k(i,j)\subseteq M(i',j')$.
This implies $M(i,j)\subseteq M(i',j')$ and, by symmetry, $M(i,j)=M(i',j')$.
In the base case of $k=0$, we have \[M_0(i,j) = ([0\dd 2\alpha_1]\cup [j-i-\alpha_1\dd j-i))\cap [0\dd j-i) = M_0(i',j').\]
Now, consider $k>0$ and $q\in M_k(i,j)$.
If $q\in B_k(i,j)\setminus [2\alpha_k\dd j-i-\alpha_k)$, then $q\in M_{k-1}(i,j)$, and thus $q\in M(i',j')$
holds by the inductive assumption.
As for the remaining case, $M_k(i,j)\cap [2\alpha_k\dd j-i-\alpha_k)=M_k(i',j')\cap [2\alpha_k\dd j'-i'-\alpha_k)$
is a direct consequence of $B_k(i,j)\cap [2\alpha_k\dd j-i-\alpha_k)=B_k(i',j')\cap [2\alpha_k\dd j'-i'-\alpha_k)$,
which follows from \cref{lemma:alphalk}.
\end{proof}

\begin{definition}\label{def:mp}
  Let $P$ be a substring of $S$ and let $S[i\dd j]$ be its arbitrary occurrence. We define $M(P):= M(i,j)$;
  by \cref{it:cons} of \cref{lem:m}, this does not depend on the choice of the occurrence.
\end{definition}

By \cref{lem:m}, the set $M(P)$ is of size $\Oh(\log \vert P\vert )$, yet, for every level $k\ge 0$ and every occurrence $P=S[i\dd j]$, it includes
the leftmost phrase boundary in $B_k(i,j)$.
Our index exploits the latter property for the largest $k$ with ${B_k(i,j)\ne \emptyset}$.

\section{Indexing with our Grammar}\label{sec:locating}

In this section, we obtain our main result by adapting the results on attractors~\cite[Sec.~6]{talg} to our modified parsing.

\begin{definition}[\cite{talg}]
The parse tree of an RLSLP is a rooted tree with leaves labeled with terminals and internal nodes labeled with nonterminals.
The root is labeled with the initial symbol of the RLSP. If a node is labeled with a nonterminal $A$ associated with a production $A \rightarrow A_1\cdots A_s$, then it has $s$ children labeled with $A_1,\ldots,A_s$, respectively, in the left-to-right order. If, instead, $A$ is associated with a production $A \rightarrow A_1^s$, then the node has $s$ children, each labeled with $A_1$.
\end{definition}

\begin{definition}[\cite{talg}]
  The grammar tree of an RLSLP is obtained by pruning its parse tree: all but the
  leftmost occurrences of each nonterminal are converted into leaves and their
  subtrees are pruned. We treat rules $A\rightarrow A_1^s$ (assumed to be of size 2) as
  $A\rightarrow A_1A_1^{[s-1]}$, where the node labeled $A_1^{[s-1]}$ is always
  a leaf ($A_1$~is also a leaf unless it is the leftmost occurrence of
  $A_1$). 
\end{definition}

Note that the grammar tree has exactly one internal node per distinct nonterminal and the total number of nodes is equal to the grammar size plus one.
We associate each nonterminal $A$ with the only internal grammar tree node labeled with~$A$.
We also sometimes associate terminal symbols with grammar tree leaves.

The search algorithm classifies the occurrences of a pattern $P$ in $S$ into
\emph{primary} and \emph{secondary}, according to the partition of $S$ induced by the grammar tree leaves.

\begin{definition}[\cite{talg}]
  The leaves of the grammar tree induce a partition of $S$ into
  phrases. An occurrence $S[t\dd t+m)$ of $P[1\dd m]$ is \emph{primary} if the lowest
  grammar tree node deriving a fragment of $S$ that contains $S[t\dd t+m)$ is
  internal (or, equivalently, the occurrence crosses the boundary between two
  phrases); otherwise, the occurrence is \emph{secondary}.
\end{definition}

The general idea of the search is to find the primary occurrences by looking for prefix-suffix partitions of $P$ and then derive the secondary occurrences from the primary ones~\cite{CNP20}.

\subsection{Finding the primary occurrences}

Let a nonterminal $A$ be the lowest (internal) grammar tree node that covers a
primary occurrence $S[t\dd t+m)$ of $P[1\dd m]$. Then, if $A\rightarrow A_1\cdots
A_s$, there exist $i \in [1\dd s)$ and $q\in [1\dd m)$ such that a
suffix of $\exp{A_i}$ matches $P[1\dd q]$ whereas a prefix of $\exp{A_{i+1}}\cdots
\exp{A_{s}}$ matches $P(q\dd m]$. The idea is to index all the pairs
$(\exp{A_i}^{rev}, \exp{A_{i+1}}\cdots \exp{A_{s}})$ and find those where the first
and second component are prefixed by $(P[1\dd q])^{rev}$ and $P(q\dd m]$,
respectively. Note that there is exactly one such pair per border between two
consecutive phrases (or leaves in the grammar tree).
On the other hand, if $A \rightarrow A_1^s$, then there exists $q \in [1\dd m)$ such that a suffix of $\exp{A_1}$ matches $P[1\dd q]$ whereas a prefix of $\exp{A_1}^{s-1}$ matches $P(q\dd m]$. We will also index the pairs $(\exp{A_1}^{rev},\exp{A_1}^{s-1})$ and find those where the first
and second component are prefixed by $(P[1\dd q])^{rev}$ and $P(q\dd m]$, respectively.

\begin{definition}[\cite{talg}]
  Let $v$ be the lowest (internal) grammar tree node that covers a primary
  occurrence $S[t\dd t+m)$ of $P$. Let $v_i$ be
  the leftmost child of $v$ that overlaps $S[t\dd t+m)$.
  We say that the node $v$ is the \emph{parent} of the primary
  occurrence $S[t\dd t+m)$ of $P$ and the node $v_i$ is its \emph{locus}.
\end{definition}

The index of~\cite{talg} builds a two-dimensional grid data structure \cite{clp:socg11}. First, two multisets of strings, $\cal X$ and $\cal Y$, are constructed. For every nonterminal $A \rightarrow A_1\cdots A_s$ and every $i \in [1\dd s)$, the strings $\exp{A_i}^{rev}$ and $\exp{A_{i+1}}\cdots \exp{A_{s}}$ are added to $\cal X$ and $\cal Y$, respectively. For every nonterminal $A \rightarrow A_1^s$, the strings $\exp{A_1}^{rev}$ and $\exp{A_1}^{s-1}$ are added to $\cal X$ and $\cal Y$, respectively. For each nonterminal $A \rightarrow A_1\cdots A_s$ and each $i \in [1\dd s)$, the grid contains a point $(x,y)$ such that $x$ is the lexicographic rank of $\exp{A_i}^{rev}$ in $\cal X$ and $y$ is the lexicographic rank of $\exp{A_{i+1}}\cdots \exp{A_{s}}$ in $\cal Y$; the point is associated with the locus $A_i$. For each nonterminal $A \rightarrow A_1^s$, the grid contains a point $(x,y)$ such that $x$ is the lexicographic rank of $\exp{A_1}^{rev}$ in $\cal X$ and $y$ is the lexicographic rank of $\exp{A_1}^{s-1}$ in $\cal Y$; the point associated with the locus $A_1$. 

Note that the number of strings in $\cal X$ and $\cal Y$, as well as the number of points in the grid, are bounded by the grammar size, $g = \Oh(\deltalos)$ in our case.

Given a partition $P = P[1\dd q] \cdot P(q\dd m]$ to look for primary occurrences, the query algorithm searches $\cal X$ for the lexicographic range $[x_1,x_2]$ of all the strings that are prefixed with $P[1\dd q]^{rev}$, and $\cal Y$ for the lexicographic range $[y_1,y_2]$ of all the strings that are prefixed with $P(q\dd m]$. By interpreting the lexicographic ranges as $x$-coordinate and $y$-coordinate ranges, respectively, it follows that the grid points within the orthogonal range $[x_1,x_2] \times [y_1,y_2]$ corresponds precisely to all the $occ_q$ primary occurrences of $P$ cut at position $q$. The geometric data structure, using $O(g)$ space, finds the points in such range of the grid in time $O((occ_q+1)\log^\epsilon g)$, where $\epsilon>0$ is an arbitrary constant fixed at construction time \cite[Thm.~2.1, using $B=\log^\epsilon n/\log\log n$]{clp:socg11}.

In order to efficiently search $\cal X$ for all the reversed prefixes $P[1\dd q]^{rev}$ and $\cal Y$ for all the suffixes $P(q\dd m]$, we make use of the following lemma.

\begin{lemma}[{\cite[Lem.~6.5]{talg}}] \label{lemma: z-fast}
      Let $\mathcal S$ be a set of strings augmented with a data
structure supporting:
\begin{itemize}
  \item extraction, in time $f_e(\ell)$, of any length-$\ell$ prefix of strings in
$\mathcal S$, and 
\item computation,  in time $f_h(\ell)$, of the Karp--Rabin
signature $\phi$ of any length-$\ell$ prefix of strings in $\mathcal S$.
\end{itemize}
We can then build a data structure of $O(\vert{\mathcal S}\vert)$ words such that,
later, we can solve the following problem in $O(m + \tau( f_h(m) +\log m ) +
f_e(m))$ time: given a pattern $P[1\dd m]$ and $\tau>0$ suffixes
$Q_1,\dots,Q_\tau$ of $P$, find the ranges of strings in (the
lexicographically-sorted) $\mathcal S$ prefixed by $Q_1,\dots,Q_\tau$.
\end{lemma}

We use this lemma for ${\cal S} = \cal X$ and ${\cal S} = \cal Y$, so the extraction of a prefix of a string in $\cal S$ corresponds to extracting the suffix of $\exp{A_i}$ for some nonterminal $A_i$ (in the case ${\cal S} = \cal X$) or to extracting a prefix of $\exp{A_{i+1}}\cdots \exp{A_s}$ or of $\exp{A_1}^{s-1}$ (in the case ${\cal S} = \cal Y$). In this second case, we need to extract the complete expansion of zero or more nonterminals, and then a prefix of the expansion of a nonterminal. It is shown in \cite{talg} that prefixes and suffixes of the expansions of nonterminals can be extracted in linear time, which yields $f_e(\ell)=O(\ell)$ in our grammar.

\begin{lemma}[{\cite[Lem.~6.6]{talg}}] \label{lem:extract from rlcfg} 
For every RLSLP of size $g$, there exists a data structure of size $O(g)$ such that, for every nonterminal $A$, any prefix or suffix of $\exp{A}$ can be obtained in time proportional to its length.
\end{lemma}

The second requirement of Lemma~\ref{lemma: z-fast} is to compute the Karp--Rabin \cite{KR87} signature $\phi$ on prefixes of strings in $\cal S$. Those are functions of the form $\phi(S) = \left(\sum_{i=1}^{\vert S\vert} S[i] \cdot b^{\vert S\vert -i}\right) \bmod p$, for some prime number $p$ and a randomly chosen integer $b \in [1 \dd p)$. To prove Lemma~\ref{lemma: z-fast}, 
the signature $\phi$ is chosen in~\cite{talg} so that there is no collision between substrings of $\cal S$, and thus the searches for the suffixes $Q_i$ are deterministically correct. In a followup to the conference version of our paper, Navarro \cite{Nav23} showed how to obtain $f_h(\ell)=O(\log\ell)$ in our grammar.

\begin{lemma}[{\cite[Lem.~12]{Nav23}}] \label{lem:kr}
The RLSLP of Section~\ref{sec:grammar} can be enhanced with $O(g)$ additional space so that the Karp--Rabin signature of 
any length-$\ell$ prefix of a string in $\mathcal{X}$ or $\mathcal{Y}$ can be extracted in time $f_h(\ell) =
O(\log \ell)$.
\end{lemma}

With this result, we obtain time $O(m+\tau\log m)$ to perform all the string searches using Lemma~\ref{lemma: z-fast}. Overall, if we have identified $\tau$ cuts of $P$ that suffice to find all of its primary occurrences in $S$, then we can find all the $occ_p \le occ$ primary occurrences of $P$ in time $O(m + \tau \log m + (occ_p+\tau) \log^\epsilon g)$.

\subsection{Parsing the pattern}

The next step is to set a bound for $\tau$ with our parsing and show how to find the corresponding cuts. As shown below, since only the leftmost and rightmost $O(1)$ phrase boundaries in each level of the parsing of $P$ can differ from those in the parsing of any occurrence of $P$ in $T$, the leftmost phrase boundary within every primary occurrence of $P$ is guaranteed to belong to $M(P)$.

\begin{lemma}
  Using our grammar of \cref{sec:grammar}, there are only $\tau = O(\log m)$
  cuts $P = P[1\dd q] \cdot P(q\dd m]$ yielding primary occurrences of $P[1\dd m]$. These positions
  $q$ belong to $M(P) + 1=\{r + 1 : r\in M(P)\}$ (see \cref{def:mp}).
\end{lemma}
\begin{proof}
  Let $A$ be the parent of a primary occurrence $S[t\dd t+m)$, and let $k$ be
  the round where $A$ is formed. There are two possibilities:

\begin{enumerate}[label=(\arabic*)]
\item{$A\rightarrow A_1\cdots A_s $ is a block-forming rule, and 
  a suffix of $\exp{A_i}$ matches $P[1\dd q]$, for some $i\in [1\dd s)$ and $q\in [1\dd m)$.
  This means that $q-1 = \min B_{k-1}(t,\allowbreak {t+m-1})$.}
\item{$A\rightarrow A_1^s$ is a run-length nonterminal and a suffix of
    $\exp{A_1}$ matches $P[1\dd q]$, for some $q\in [1\dd m)$. This means that $q-1
    = \min B_{k-1}(t, \allowbreak {t+m-1})$.}
\end{enumerate}
  In either case, $q \in M(P)+1$ by \cref{lem:m}. Further, $\vert M(P)\vert  = O(\log m)$.
\end{proof}

The parsing of $P$ is performed done in $O(m)$ time much as in previous work~\cite[Sec.~6.1]{talg}, with the difference that we have to care about paused symbols. Essentially, we store the permutations $\pi_k$ drawn when indexing $S$ and use them to parse $P$ in the same way, round by round, aiming to create the same symbols. Actually, we parse the string $P^* = \#P\$$, where \# and \$ are two symbols that appear neither in $S$ nor in $P$. We use the occurrence of $P$ in $P^*$ to build $M(P)$.

As observed in the proof of \cref{corollary:deltalog}, we have $\sum_{k=0}^\kappa \vert \S_k \cap \A_{k+1}\vert  = O(\deltalos)$. Hence, we store the values of $\pi_{k+1}$ only for the active symbols in $S_k$ for every $k\in [0\dd \kappa]$;
the values $\pi_{k+1}$ for the remaining symbols do not affect the placement of block boundaries in \cref{def:block}:
If $S_k[j],S_{k}[j+1]\in \A_{k+1}$, then, due to the condition imposed on $\pi_{k+1}$ in \cref{construction},
$j$ may only be a local minimum if $S_{k}[j-1]\in \A_{k+1}$.
When parsing $P$, we can simply assume that $\pi_{k+1}(A)=0$ on the paused symbols $A\in \S_k\setminus \A_{k+1}$ and obtain the same parsing of $S$. By storing the values of $\pi_k$ only for the active symbols, we use $\Oh(\deltalos)$ total space. 

In order to find the correct nonterminals among those used in the parsing of $S$, we store two perfect hash tables \cite{fks:jacm84}. In the first one, we keep all the keys $(A_1,s)$, with associated information $A$, corresponding to the grammar rules of the form $A \rightarrow A_1^s$. In the second one, we keep, for each rule $A \rightarrow A_1 \cdots A_s$ formed in $S$, the key $(A_1 \cdots A_s)$ with associated information $A$. The size of the hash tables is proportional to the size $g$ of
the grammar.

Let us describe the first two rounds of the parsing \cite{talg}. We first traverse $P^* = P^*_0$ left-to-right and identify the runs $a^\ell$ in $O(m)$ time. For each such run, we search the first hash table for the key $(a,\ell)$, finding its corresponding nonterminal $A$, with which we replace the run (see below for the case where the run does not appear in the hash table). The result of this pass is a new sequence $P^*_1$. We then traverse $P^*_1$, finding the local minima according to $\pi_1$, and thus identifying the blocks. To do this in $O(m)$ time as well, we search the second hash table for the identified blocks, replacing them by the corresponding nonterminals, and forming the sequence $P^*_2$.

By \cref{cor:sizeSk}, the number of phrases in round $k$ is less than $1 + \frac{4m}{\ell_{k+1}}$, which gives us at most $h = 12 + 2\lfloor{\log_{4/3}m}\rfloor = O(\log m)$ parsing rounds and a total of $\sum_{k=0}^{h} (1 + \frac{4m}{\ell_{k+1}}) = O(m)$ symbols processed along the parsing of $P$. Since we spend constant time per symbol along the parse the whole parsing takes time $O(m)$. Construction of the set $M(P)$ from Definitions~\ref{def:M} and \ref{def:mp} (along with the auxiliary sets $B_k(i,j)$ and $M(i,j)$) also takes $O(m)$ time.

Note that $P^*_k$ may contain runs and blocks that do not occur in $S_k$. By Lemma~\ref{lemma:alphalk}, other than those overlapping the first $2\alpha$ or the last $\alpha$ positions of of $P$, where $\alpha := \lfloor 8\ell_k\rfloor$, any run or block formed inside $P^*$ must also be formed inside every occurrence of $P$ in $S$. Consequently, if a run or a block does not overlap those extreme positions, yet it does not appear in the hash table, we can abandon the search because $P$ cannot occur in $S$. On the other hand, by Lemma~\ref{lemma:BkI}, there can be only $O(1)$ symbols overlapping those areas in each level $P^*_k$, and thus $O(\log m)$ in total. We can then 
gather those unknown runs and blocks appearing in the extremes of $P^*$ in order to consistently assign them new nonterminals and arbitrary unused $\pi_k$ values. We then proceed normally with subsequent levels of the parsing. Traversing the set linearly in order to detect if they reappear in the parsing adds up to just $O(\log^2 m)$ total time.

\subsection{Secondary occurrences and short patterns}\label{subsec:secondary}

Thus far, we have obtained the loci of the primary occurrences, but not yet their positions in $S$. Further, we must find the secondary occurrences that derive from the primary ones. In order to find those, we use a technique that works for any arbitrary RLSLP and within $O(g)$ space~\cite[Sec.~6.4]{talg}. 
For each grammar tree node $v$ labeled $A$, we store (i) $u = v.\mathit{anc}$, the nearest ancestor of $v$, labeled $B$, such that $u$ is the root or $B$ labels more than one node in the grammar tree; (ii) $v.\mathit{offs}$, the offset at which $v$ (i.e., $\exp{A}$) starts inside $u$ (i.e., $\exp{B}$); and (iii) $v.\mathit{next}$, the next node in preorder traversal of the grammar tree labeled $A$, or $null$ if $v$ is the last node labeled $A$. The first two fields are undefined for the nodes labeled $A^{[s-1]}$ that we create in the grammar tree (these nodes do not exist in the parse tree), whereas these nodes participate in the $v.next$ list of label $A$.

For each primary occurrence found with the partition $P = P[1\dd q] \cdot P(q\dd m]$, the point we found in the grid gives us its locus $v$ in the grammar tree. Say $v$ is labeled $A_i$, and its parent $u$ is labeled $A$, with $A \rightarrow A_1
\cdots A_s$. For each such node $v$, we store $\vert \exp{A_1\cdots A_i}\vert$, so the offset of $P$ within $\exp{A}$ is $o := \vert \exp{A_1\cdots A_i}\vert-q+1$. To find the offset within $S$ of the primary occurrence, we repeatedly traverse pointers from $u$ to $u.\mathit{anc}$ and update $o := o + u.\mathit{offs}$, until reaching the root. 
At this point, we can report an occurrence $S[o \dd o+m)$ of $P$.

In our upwards way to the root, we will also report secondary occurrences. For every node $u$ we visit in an upward path, we recursively continue not only by $u.\mathit{anc}$ but also by $u.\mathit{next}$; the offset $o$ is retained in this second case. The tree of recursive calls is binary and reports a different occurrence of $P$ in $S$ at each leaf reaching the initial symbol. Therefore, the positions of all the occurrences, primary and secondary, are reported in constant amortized time each.\footnote{The unproductive tests $u.\mathit{next} = \mathit{null}$ are charged to the primary occurrence $v$ if the label of $u$ is $A$, or to the first secondary occurrence of $u$, which exists by the definition of $u.\mathit{anc}$, otherwise.}

The case where $i=1$ and the parent $u$ of $v$ is labeled $A \rightarrow A_1^s$ is special. We compute $o := \vert\exp{A_1}\vert-q+1$ and report occurrences in $u$ with offsets $o+i\cdot\vert\exp{A_1}\vert$ for $i = 0,1,\ldots$ as long as $o+i\cdot\vert\exp{A_1}\vert+m-1 \le s\cdot \vert\exp{A_1}\vert$, apart from the call to $v.\mathit{next}.\mathit{next}$. On the other hand, the nodes $u$ labeled $A_1^{[s-1]}$ reached by pointers $\mathit{next}$ along the recursive calls are treated exactly as $s-1$ copies of $A_1$, triggering $s-1$ calls to $u.\mathit{anc}$ with offsets $o+i\cdot\vert\exp{A_1}\vert$ for $i \in [1\dd s)$, and also one call to $u.\mathit{next}$ with offset $o$.

Plugged with the preceding results, the total space of our index is $\Oh(\deltalos)$ and its search time is $O(m + \tau(\log^\epsilon g + \log m) + occ \log^\epsilon g) = O(m + (\log m +occ )\log^\epsilon g)$. This bound exceeds $O(m + ({occ+1})\log^\epsilon g)$ 
only when
$m=O(\log^\epsilon g\log\log g)$. In that case, however, $O(\log m \log^\epsilon g) = O(\log\log g\cdot \log^\epsilon g)$, which becomes $O(\log^\epsilon g)$ again if we infinitesimally adjust~$\epsilon$. 
Our complexity then becomes $O(m+\log^\epsilon g + occ\log^\epsilon g)$.

The final touch is to reduce that complexity to $O(m + \log^\epsilon \delta + occ \log^\epsilon g)$. This is relevant only when $occ=0$, so we need a way to detect in time $O(m+\log^\epsilon \delta)$ that $P$ does not occur in $S$. We already do this in time $O(m+\log^\epsilon g)$ by parsing $P$ and searching for its cuts in the geometric data structure. To reduce the time, we note that $\log^\epsilon g
\le \log^\epsilon \delta + \log^\epsilon \frac{g}{\delta}$, so it suffices to detect in $O(m)$ time the patterns of length $m \le \ell := \log^\epsilon \frac{ g}{\delta}$ that do not occur in~$S$. By definition of $\delta$, there are at most $\delta\ell^2$ substrings of length at most $\ell$ in $S$, so we can store them all in an (uncompressed) trie using total space $O(\delta\ell^2)=\Oh(\delta \log^{2\epsilon} \frac{g}{\delta})=\Oh(g)=\Oh(\deltalos)$. By implementing the trie children with perfect hashing \cite{fks:jacm84}, we can verify in $O(m)$ time whether a pattern of length $m \le \ell$ occurs in $S$. We then obtain Theorem~\ref{thm:main}.

\subsection{Optimal search time}\label{subsec:optimal-search}

We now show how to obtain optimal search time, $O(m+occ)$, by slightly increasing the size of our $\delta$-bounded index. As in previous work \cite[Sec.~6.7]{talg}, we use a larger geometric structure, of size $O(g \log^{\epsilon} g)$ for some constant $\epsilon > 0$, which reports in $O(\log \log g)$ time per query and $O(1)$ per result. This structure is constructed in $O(g \log g)$ expected time using $O(g \log^{\epsilon} g)$ working space \cite[Thm.~2]{abr:focs2000}.

With the enhanced geometric structure, the query time automatically drops to $O(m + \log m \log \log g + occ)$. This can be written as $O(m + \ell + occ)$, with $\ell = \log \log g \log \log \log g$~\cite[Sec.~6.7]{talg}, because the middle term dominates only if $m \le \log m\log\log g$.

In order to achieve the optimal search time, we need only to care about patterns of length up to $\ell$ that occur less than $\ell$ times, since otherwise the term $occ$ absorbs $\ell$. To do this, we store all the text substrings of length up to $\ell$ in an uncompressed trie $C$. This trie stores the number of times each of its nodes occurs in the text, and stores in its leaves at depth $\ell$ the actual list of those occurrences, in case there are at most $\ell$. By definition of $\delta$, there are at most $\delta \ell$ distinct substrings of length $\ell$, so the trie contains at most $\delta\ell^2$ nodes.
Moreover, across all the leaves, we store at most $\delta\ell^2$ occurrences. The total space for $C$ is then $O(\delta\ell^2) \subseteq \Oh(\delta\log^\epsilon g)$ for any constant $\epsilon>0$.

To search for a pattern $P$ of length up to $\ell$, we first search for it in $C$ and verify if its trie node indicates that it occurs more than $\ell$ times. If so, we use the normal search, as described above, which on the enhanced grid takes time $O(m+\ell+occ) = O(m +occ)$.
Otherwise, its occurrences are collected from all the leaves descending from its trie node, also in time $O(m + occ)$. Instead, if the pattern length exceeds $\ell$, it is also searched normally, in time $O(m+\ell+occ) = O(m+occ)$.

To build $C$, we slide a window of length $\ell$ through the whole text $S$, maintaining a Karp--Rabin fingerprint $\phi$ for the current window. We store the distinct signatures $\phi$ found in a hash table $H$, with a counter of its corresponding  length-$\ell$ substring in $S$. When a new $\phi$ value appears, the underlying string is inserted in $C$, a pointer from $H$ is set to point to the corresponding trie leaf, and the list of the occurrences of the substring is initialized in the corresponding leaf, with its first position just found. Every new occurrence found is added in its trie leaf, until the length of its list exceeds $\ell$, in which case the list is deleted and not maintained anymore.

Per the Karp--Rabin formula given above, it takes $O(1)$ time to compute $\phi(S\cdot b)$ from $\phi(a \cdot S)$; thus, it takes $O(n)$ time to compute the fingerprints for all the length-$\ell$ windows of $S$. 
It takes $O(n)$ expected time to maintain $H$ (also considering collisions) and collect the occurrences, plus $O(\delta \ell^2)$ time to insert strings in $C$. We can propagate the number of occurrences of each leaf upwards in the trie, in $O(\delta \ell)$ additional time. Added to the $O(g\log g)$ expected time to build the enhanced grid, the expected construction time is $O(n+g\log g) \subseteq \Oh(n + \deltalos\log n)$. This completes the proof of Theorem~\ref{thm:main2}.

\section{Counting pattern occurrences}
In this section, we show how to efficiently count pattern occurrences within $\delta$-bounded space
by adapting previous results \cite[Sec.~7]{talg} to our parsing.

\subsection{Near-optimal counting time}

Christiansen et al.~\cite{talg} build an RLSLP of size $O(\gammalog)$ on $S[1\dd n]$ to count how many times a pattern $P[1\dd m]$ occurs in $S$ in the near-optimal time $O(m + \log^{2+\epsilon} n)$, for any constant $\epsilon>0$. The idea is that, instead of tracking all the secondary occurrences that are reached from each point of the grid that lies in the appropriate ranges (recall \cref{subsec:secondary}), the number of those occurrences is stored with the corresponding grid point, so they just need to sum up those numbers. The only problem is with the run-length rules, as in those cases the number of occurrences is not a function of the points only. This case is handled using properties of their particular parsing.

Except for the problem of the run-length rules, the idea works on any grammar. As seen in \cref{sec:locating}, using our grammar of size $g=\Oh(\deltalos)$, there are $O(\log m)$ cutting positions that must be tried when searching for $P$. We need $O(m)$ time to identify the $O(\log m)$ relevant two-dimensional ranges,  and then count the number of occurrences in each such range in time $O(\log^{2 + \epsilon} g)$, using the same geometric data structure as in previous work~\cite{talg}. Our counting time is then $O(m + \log m \log^{2 + \epsilon}g)$. If $m \leq \log m \log^{2 + \epsilon}g$, then we have $\log m \log^{2 + \epsilon}g \in O(\log^{2 + \epsilon}g \log \log g)$. In that case, infinitesimally adjusting $\epsilon$, we have $\log m \log^{2 + \epsilon}g \in O(\log^{2+\epsilon}g)$; thus, our counting time can be written as $O(m + \log^{2 + \epsilon} g)$.

All that is needed is then to handle the problem of the run-length rules in our particular grammar. Just as in previous work~\cite{talg}, we will rely on a property of the 
 shortest period of the expansion of the string generated by a run-length rule. In our grammar, we must take care of the paused symbols.

\begin{definition}
A string $P[1\dd m]$ has a period $p$ if $P$ consists of $\lfloor m/p \rfloor$ consecutive copies of $P[1\dd p]$ plus a (possible empty) prefix of $P[1\dd p]$.
\end{definition}

The only property required for the mechanism designed for run-length rules~\cite{talg} is the one we prove in the following lemma.

\begin{lemma}[{cf.~\cite[Lemma 7.2]{talg}}]\label{lemma:period}
For every run-length rule $A \rightarrow A_1^s$ in our grammar, the shortest period of $\exp{A}$ is $\vert A_1 \vert := \vert \exp{A_1}\vert $.
\end{lemma}
\begin{proof}
Let $\exp{A} = S(i\dd j]$. Since $\exp{A}$ consists of $s$ consecutive copies of $\exp{A_1}$, then $\vert A_1\vert $ is a period of $\exp{A}$.
From the Periodicity Lemma~\cite{periodicity}, we know that the shortest period $p$ of $\exp{A}$ satisfies $p = \gcd{(p, \vert A_1\vert )}$; thus, $d := \vert A_1\vert /p$ is an integer.
Let $v$ be the node labeled $A$ in the parse tree of the grammar, and let $r + 1$ denote the level of the run represented by $v$, so $A$ is a symbol in $S_{r+1}$ and $A_1$ is a symbol in $S_{r}$.
\begin{claim}\label{claim:same-blocks}
For every $k \le r$, there are $ds$ identical fragments in $S_k$ that expand to the (identical) fragments $S(i \dd i + p], S(i+p \dd i+2p], \ldots, S(j - p  \dd j]$.
\end{claim}
\begin{proof}
    Since $p$ is a period of $S(i\dd j]=S_0(i\dd j]$, we have $S_0(i \dd i + p] = \cdots = S_0(j - p \dd j]$.
    Hence, the claim holds trivially for $k = 0$. For $k > 0$, the inductive assumption yields $ds$ identical fragments $S_{k-1}(i_k\dd i_k+p_k],\ldots, S_{k-1}(j_k-p_k\dd j_k]$ that expand to $S(i\dd i+p],\ldots,S(j-p\dd j]$, respectively.
    Since $i,i+dp,j$ are all level-$r$ phrase boundaries, they are also level-$k$ phrase boundaries. Consequently, there is a block boundary after $S_{k-1}[i_k]$, $S_{k-1}[i_k+dp_k]$, and $S_{k-1}[j_k]$. In particular, $S_{k-1}(i_k\dd j_k]$ is parsed into full blocks, which are collapsed  to $S_k(i_{k+1}\dd j_{k+1}]$. We shall prove that, for every $\delta\in [0\dd p_k)$ and $t\in [0\dd ds)$, there is a block boundary after $S_{k-1}[i_k+tp_k+\delta]$ if and only if there is one after $S_{k-1}[i_k+dp_k+\delta]$. For this, we consider several cases.
    \begin{itemize}
      \item If $tp_k+\delta=0$, there are block boundaries after both $S_{k-1}[i_k]$ and $S_{k-1}[i_k+dp_k]$.
      \item If $k$ is odd and $tp_k+\delta \ge 1$, then
      \[S_{k-1}[i_k+tp_k+\delta\dd i_k+tp_k+\delta+1]=S_{k-1}[i_k+dp_k+\delta\dd i_k+dp_k+\delta+1].\] Since run-length encoding places a block boundary after $S_{k-1}[\ell]$ solely based on $S_{k-1}[\ell\dd \ell+1]$, we conclude that there is a block boundary after $S_{k-1}[i_k+tp_k+\delta]$ if and only if there is one after $S_{k-1}[i_k+dp_k+\delta]$.
      \item If $k$ is even, $tp_k+\delta\ge 2$, and $p_k\ge 2$, then \[S_{k-1}[i_k+tp_k+\delta-1\dd i_k+tp_k+\delta+1]=S_{k-1}[i_k+dp_k+\delta-1\dd i_k+dp_k+\delta+1].\] Since the restricted block parsing places a block boundary after $S_{k-1}[\ell]$ solely based on $S_{k-1}[\ell-1\dd \ell+1]$, we conclude that there is a block boundary after $S_{k-1}[i_k+tp_k+\delta]$ if and only if there is one after $S_{k-1}[i_k+dp_k+\delta]$.
      \item If $k$ is even, $tp_k+\delta=1$, and $p_k\ge 2$, then we need to consider several possibilities. If either $S_{k-1}[i_k+1]\notin \A_k$ or $S_{k-1}[i_{k}+2]\notin \A_k$, then, due to $S_{k-1}[i_k+1\dd i_k+2]=S_{k-1}[i_k+dp_k+1\dd i_k+dp_k+2]$, restricted block parsing places block boundaries both after $S_{k-1}[i_k+1]$ and $S_{k-1}[i_k+dp_k+1]$. Thus, we can assume that both $S_{k-1}[i_k+1]\in \A_k$ and $S_{k-1}[i_k+2]\in \A_k$.
      Since there is a block boundary after $S_{k-1}[i_k]$, then $i_k=0$ or $\pi_{k}(S_{k-1}[i_k])<\pi_{k}(S_{k-1}[i_k+1])$; the latter holds even if $S_{k-1}[i_k]\notin \A_k$ because $\pi_k$ assigns larger values to symbols in $\A_k$ than to symbols outside $\A_k$ (see~\cref{construction}).
      In either case, this means that $\pi_{k}(S_{k-1}[i_k+1])$ is not a local minimum, so restricted block parsing does not place a block boundary after $S_{k-1}[i_k+1]$.
      Similarly, since there is a block boundary after $S_{k-1}[i_k+dp_k]$, then $\pi_{k}(S_{k-1}[i_k+dp_k])<\pi_{k}(S_{k-1}[i_k+dp_k+1])$,
      so there is no block boundary after $S_{k-1}[i_k+dp_k+1]$.
      \item If $k$ is even and $p_k=1$, then all the symbols of $S_{k-1}(i_k\dd j_k]$ are equal. The run they constitute has not collapsed at level $k-1$,
      so the underlying symbol does not belong to $\A_{k-1}=\A_k$. Consequently, all symbols of $S_{k-1}(i_k\dd j_k]$ form length-1 blocks. In particular,
      there are block boundaries after both $S_{k-1}[i_k+t]$ and $S_{k-1}[i_k+d]$.
    \end{itemize}
    All in all, we have proved that all fragments $S_{k-1}(i_k+tp_k\dd i_k+(t+1)p_k]$ consist of full blocks and that they are parsed identically.
    Hence, each of them gets collapsed to an identical fragment of $S_{k}$ that expands to $S(i+tp\dd i+(t+1)p]$.
\end{proof}
Applying \cref{claim:same-blocks} to $k=r$, we conclude that $i,i + p, i + 2p, \ldots, j$ are all level-$r$ phrase boundaries.
In particular, $S(i\dd j]$ consists of at least $ds$ level-$r$ phrases. However, $S(i\dd j]$ constitutes exactly one level-$(r+1)$ phrase corresponding to $A$. Since $A \rightarrow A_1^s$, we conclude that $ds=s$, i.e., $p=\vert A_1\vert$ holds as claimed.
\end{proof}
 By \cref{lemma:period}, we can safely use the original method~\cite{talg} to handle the run-length rules. Thus, using $\Oh(\deltalos)$ space, we can build in $O(n\log n)$ expected time a data structure that can count the occurrences of a pattern $P[1 \dd m]$ within $S[1 \dd n]$ in near-optimal time $O(m + \log^{2+\epsilon}n)$. This proves the first part of Theorem~\ref{thm:main3}.

\subsection{Optimal counting time}
We now show how to reduce the counting time offered by our index to the optimal $O(m)$, at the cost of increasing the space consumption by a factor of $O(\log n)$. If $\log g \ge \frac{n\log \sigma}{\delta \log n}$, then we can simply use a compact index~\cite{DBLP:journals/talg/BelazzouguiN14},
which takes $\Oh(\frac{n\log \sigma}{\log n})\subseteq\Oh(\delta \log g)=\Oh(\delta\log(\deltalos))$ space, can be constructed in $\Oh(n)$ expected time, and answers counting queries in $\Oh(m)$ time. Thus, we henceforth assume that $\log g < \frac{n\log \sigma}{\delta \log n}$.

In the same way as done in previous work \cite[Sec.~7.1]{talg}, we use a larger geometric structure that speeds up counting time. This structure \cite[Thm.~3]{abr:focs2000} supports orthogonal range counting queries in $O(\log g)$ time using $O(g \log g)$ space, and is built in $O(g\log^2 g)$ expected time. The counting time of our index then becomes $O(m+\log m \log g)$.

Using a compact trie similar to the one described in \cref{subsec:optimal-search}, we index all the substrings of length at most $\ell = \ceil{\log g \log \log g}$. The only difference with the trie we used before is
that the trie we use now is compact (i.e., internal nodes with one child are not stored explicitly) and that we do not store the list of occurrences but just the number of occurrences at each node. Since there are $\delta\ell$ substrings of length exactly $\ell$, the space consumption of the compact trie is $\Oh(\delta \ell) =\Oh(\delta \log g \log \log g)$.

Our counting time $O(m + \log m \log g)$ is in $O(m)$  when $m > \ell$ because then $
\log g = \Oh(\frac{\ell}{\log \ell})\subseteq\Oh(\frac{m}{\log m})$. If we answer queries of length up to $\ell$ using the compact trie and any other query using the geometric structure, our counting time is always $O(m)$.

The space consumption of our data structure is then $\Oh(g \log g + \delta \log g \log \log g)$. 
The first term is $\Oh(\deltalos \log (\deltalos))$ because $g = \Oh(\deltalos)$.
Due to our assumption $\log g < \frac{n\log \sigma}{\delta \log n}$, the second term is $\Oh(\delta \log g \log \log g)\subseteq \Oh(\deltalos \log g)\subseteq \Oh(\deltalos \log(\deltalos))$.
The total expected construction time is $O(n\log n)$~\cite{talg}. This completes the proof of Theorem~\ref{thm:main3}.

\section{Conclusions and Future Work}\label{sec:concl}

We have obtained the best of two worlds~\cite{talg,delta} in repetitive text indexing: an index of asymptotically optimal size, $O(\deltalos)$, with nearly-optimal search time, $O(m+(occ+1)\log^\epsilon n)$, where $n$ is the text size, $\delta$ its repetitiveness measure, $m$ is the pattern length, $occ$ is the number of times the pattern occurs in the text, and $\epsilon>0$ is any constant. This closes a question open in those previous works. We also manage to search in optimal $O(m+occ)$ time within near-optimal space, $\Oh(\deltalos\log^{\epsilon}(\deltalos))$. Finally, we show how to count the occurrences of the pattern in the text without enumerating them: we can do this in $O(m+\log^{2+\epsilon} n)$ time and optimal $\Oh(\deltalos)$ space, or in optimal $O(m)$ time and $\Oh(\deltalos\log (\deltalos))$ space.

An important open question is whether this final gap can be closed: can we search or count in optimal time within $\delta$-optimal space, $\Oh(\deltalos)$? With our best results, space or time must be multiplied by $O(\log^\epsilon n)$ or $O(\log n)$.
Further improvements of the search time could also be possible in the packed setting, where the pattern can be read in $O(\ceil{\frac{m\log \sigma}{\log n}})$ time.

Further research goes in the direction of providing more complex search functionality within $\delta$-bounded space. 
For example, a recent article~\cite{Nav23} shows how to efficiently find the maximal exact matches (MEMs) of a pattern in the text, within $\Oh(\deltalog)$ space.
The way towards a final goal like providing suffix tree functionality in $\delta$-bounded space has several intermediate problems like this one, which is of relevance in bioinformatics applications.
The very recent $\delta$-SA~\cite{deltaSA}, developed after the submission of this work, supports polylogarithmic-time suffix array and inverse suffix array queries in $\Oh(\deltalos)$ space, and thus constitutes a major step in that direction.

\section*{Declarations}

\subsection*{Competing interests}
All authors certify that they have no affiliations with or involvement in any organization or entity with any financial interest or non-financial interest in the subject, matter, or materials discussed in this manuscript.

\bibliography{ref}

\end{document}